\documentclass[prodmode,acmtocl]{acmsmall} 

\usepackage[ruled]{algorithm2e}

\SetAlFnt{\small}
\SetAlCapFnt{\small}
\SetAlCapNameFnt{\small}
\SetAlCapHSkip{0pt}
\IncMargin{-\parindent}

\usepackage{amsmath,amssymb}
\usepackage[all]{xy}

\setcopyright{rightsretained}

\newcommand{\lra}{\overset{\ast}{\longleftrightarrow}}

\newcommand{\leftrightover}[1]{{^{\underleftrightarrow{#1}}}}

\newcommand{\leftover}[1]{{^{\underleftarrow{#1}}}}
\newcommand{\rightover}[1]{{^{\underrightarrow{#1}}}}

\newcommand{\rwtermss}{\mathcal{T}(\mathcal{F},X)}

\newcommand{\constants}{\mathcal{F}_0}
\newcommand{\NF}{NF_{R}}
\newcommand{\UNeq}{$\operatorname{UN}^=$}

\begin{document}

\markboth{Nicholas Radcliffe et al.}{Uniqueness of Normal Forms for
  Shallow Term Rewrite Systems}

\title{Uniqueness of Normal Forms for Shallow Term Rewrite Systems}
\author{NICHOLAS R. RADCLIFFE, LUIS F. T. MORAES, AND RAKESH M. VERMA \affil{University of Houston}}


\begin{abstract}
\noindent Uniqueness of normal forms (\UNeq) is an important property of term rewrite
systems. \UNeq{} is decidable for ground (i.e., variable-free) systems and undecidable in general. Recently
it was shown to be decidable for linear, shallow systems. We generalize this previous
result and show that this property is decidable for shallow rewrite
systems, in contrast to confluence, reachability and other properties, which are
all undecidable for flat systems. Our result is also optimal in some
sense, since we prove that the \UNeq{} property is undecidable for two
classes of linear rewrite systems: left-flat systems in which right-hand
sides are of depth at most two and right-flat systems in which
left-hand sides are of depth at most two.  
\end{abstract}

%


%
%



\acmformat{Nicholas Radcliffe, Luis F. T. Moraes, and Rakesh Verma,
  20xx. Decidability of Unicity for Shallow Term Rewrite Systems.}

\begin{bottomstuff}
A preliminary version of this paper appeared in IARCS Annual
Conference on Foundations of Software Technology and Theoretical
Computer Science (FSTTCS 2010).

Author's addresses: Luis F. T. Moraes {and} Rakesh Verma, Computer Science Department,
University of Houston.
\end{bottomstuff}

\maketitle

\section{Introduction}\label{sec-intro}

Term rewrite systems (TRSs), finite sets of rules, are useful in many
computer science fields including theorem proving, rule-based
programming, and symbolic computation. An important property of TRSs
is confluence (also known as the Church-Rosser property), which
implies unicity or uniqueness of normal forms (\UNeq{}). Normal forms are
expressions to which no rule is applicable. A TRS has the \UNeq{}
property if  
there are \emph{not} distinct normal forms $n$, $m$ such that $n
~\leftrightover{*}_R~ m$, where $\leftrightover{*}_R$ is the symmetric
closure of the rewrite relation induced by the TRS $R$. There is a
related property called $UN^\rightarrow$, which is defined as: no term should
have more than one normal form, i.e., if $m$ and $n$ are two normal
forms reachable from the same term ($\leftover{*}\circ \rightover{*}$), then $R$ does not have the
$UN^\rightarrow$ property.  This property is known to be undecidable for flat
systems and also flat and right-linear systems \cite{godoyJ09}.     

Uniqueness of
normal forms is an interesting property in itself and
well-studied \cite{Terese}.  Confluence can be too strong a
requirement for some applications such as lazy programming.
Additionally, in the proof-by-consistency approach for inductive
theorem proving, consistency is often ensured by requiring the \UNeq{}
property.

We study the decidability of uniqueness of normal
forms.  Uniqueness of normal forms is decidable for ground systems
\cite{vrl}, but is undecidable in general \cite{vrl}.  Since
the property is undecidable in general, we would like to know for
which classes of rewrite systems, beyond ground systems, we can decide \UNeq{}. In 
\cite{zv06,julian} a polynomial time algorithm for this property was given
for linear, shallow rewrite systems. A rewrite system is {\it linear} if
variables occur at most once in each side of any rule.  It 
is {\it shallow} if variables occur only at depth zero or depth one
in each side of any rule.  It is {\it flat} if both the left- and
right-hand sides of all the rules have  
height zero or one. An example of a linear flat (in fact, ground) system that has \UNeq{} but not confluence
is $\{f(c) \to 1, \, c \to g(c)\}$.  More sophisticated examples can
be constructed using a sequential `or' function in which the second
argument gives rise to a nonterminating computation.

In this
paper, we consider the class of shallow systems, i.e., we drop the
linearity restriction of \cite{zv06}, and a subset
of this class, the flat systems. For flat systems many properties
are known to be undecidable including confluence, reachability, joinability, and
existence of normal forms \cite{moj06,rakrule08,gh09}. On the other hand, the word problem is known to be
decidable for shallow systems \cite{CHJ}.  This paper shows that
the uniqueness of normal forms problem is
decidable for the class of shallow term rewrite systems, which is
a significant generalization of \cite{zv06} and also somewhat
surprising since so many properties are undecidable for this class of
systems. We also prove the undecidability of \UNeq{} for two
subclasses of linear
systems: left-hand sides are flat and right-hand sides are
of depth at most two and conversely right-flat and depth two left-hand
sides, which improves the undecidability  result of
\cite{rakrule08} for the linear, depth-two subclass and shows that
our result is optimal as far as linearity and depth restrictions are
involved. 

We would like to clarify the relationship between $UN^\rightarrow$
(see \cite{Terese} for a definition) and
$UN^=$. It is well known in rewriting that $UN^=$ implies
$UN^\rightarrow$ but not the other way around. For a simple example,
well-known since \cite{klop,klop92}, consider $a \rightarrow b$, $a \rightarrow
c$, $c \rightarrow c$, $d \rightarrow c$, and $d \rightarrow e$. This
example has $UN^\rightarrow$ since $c$ is not a normal form but does
not have $UN^=$ since normal forms $b$ and $e$ satisfy $b =_R e$, so
$UN^{\rightarrow}$ does not imply $UN^=$. However, just because
property A implies property B it does not automatically follow that if
A is decidable for a class of inputs, then B is also decidable for the
same class of inputs. For this we need the concept of a reduction and
in fact the second author has shown  ~\cite{verma09fi} that for
variable-preserving rewrite systems $UN^=$ reduces to
$UN^\rightarrow$.

{\em Comparison with related work.} Viewed at a very high level, the
proof of decidability shows some flavor in common with that of some
other decidability proofs of properties of rewrite systems such as 
~\cite{gtv03}. The basic insight is that,  just as in algebra the
terms that reduce to 0 are crucial in a sense, so in rewriting are the
terms that reduce to (or are equivalent to)  constants. We see a
parallel between constants, which are height 0 terms in rewriting with
the expression 0 in algebra. Of course, this observation is about as
helpful in proofs of decidability as a compass is to someone lost in a
maze. The details in both scenarios are vital and there are many
twists and turns. The proof of undecidability shows some similarity
with proofs in ~\cite{vrl,gt07}.  

The structure of our decidability proof is as follows: in \cite{zv06,julian}
it was shown that \UNeq{} for
shallow systems can be reduced to \UNeq{} for flat systems, (ii)
checking \UNeq{} for flat systems can be reduced to searching for
equational proofs between terms drawn from a finite set of terms, and (iii)
existence of equational proofs between  terms in 
part (ii) is done thanks to the decidability of the word problem~
by Comon et al. \cite{CHJ}. 

Our strategy for part (ii) above, assuming a flat TRS, $R$, is to show that a
sufficiently small witness to non-$UN^=$ for $R$ exists if, and only
if, any witness at all exists. To see this, say
$\langle M,N \rangle$ is a minimal witness to non-$UN^=$ (in that the
sum of the sizes of $M$ and $N$ is minimal). We show that we can
replace certain subterms of $M$ and $N$ that are not equivalent to
constants with variables, obtaining a witness $\langle M',N'
\rangle$. If the heights of $M'$ and $N'$ are both strictly less than
$\max(1,C)$, where $C$ is the number of constants in our
rewrite system, then $\langle M',N' \rangle$ is sufficiently
small. Otherwise, $M'$ or $N'$ must have a big subterm (i.e. a subterm
whose height is greater than, or equal to, the number of constants),
and this subterm is equivalent to a constant. However, in this case
(when there is a constant that is equivalent to a big subterm of a
component of a minimal witness), we can show that there is a small
witness to non-$UN^=$. So, in all cases, we end up with a small
witness. 

This paper improves our previous work in \cite{nv10} by
strengthening the undecidability proof. In particular, the previous
proofs work only for either left-nonlinear or right-nonlinear systems, whereas here the reductions
give {\em linear} systems of the appropriate type. 
\subsection{Definitions}
\noindent {\bf Terms.} A \emph{signature} is a set $\mathcal{F}$ along
with a function \emph{arity}$: \mathcal{F} \rightarrow
\mathbb{N}$. Members of $\mathcal{F}$ are called \emph{function
  symbols}, and $arity(f)$ is called the $arity$ of the function
symbol $f$. Function symbols of arity zero are called $constants$. Let
$X$ be a countable set disjoint from $\mathcal{F}$ that we shall call
the set of $variables$. The set $\mathcal{T}(\mathcal{F},X)$ of
\emph{$\mathcal{F}$-terms over $X$} is defined to be the smallest set
that contains $X$ and has the property that $f(t_1,\ldots,t_n) \in
\mathcal{T}(\mathcal{F},X)$ whenever $f \in \mathcal{F}$, $n =
arity(f)$, and $t_1,\ldots,t_n \in \mathcal{T}(\mathcal{F},X)$. The
set of function symbols with arity $n$ is denoted by $\mathcal{F}_n$;
in particular, the set of constants is denoted
by $\constants$. We use $root(t)$ to refer to the outermost function
symbol of $t$.  

The $size$, $|t|$, of a term $t$ is the number of occurrences of
constants, variables and function symbols in $t$. So, $|t| = 1$ if $t$
is a constant or a variable, and $|t| = 1 + \Sigma_{i=1}^n |t_i|$ if
$t = f(t_1,\ldots,t_n)$ for $n > 0$. The $height$ of a term $t$ is $0$ if $t$ is a
constant or a variable, and $1 + max
\{height(t_1),\ldots,height(t_n)\}$ if $t = f(t_1,\ldots,t_n)$. If a
term $t$ has height zero or one, then it is called $flat$. A
$position$ of a term $t$ is a sequence of natural numbers that is used
to identify the locations of subterms of $t$. The subterm of $t =
f(t_0,\ldots,t_{n-1})$ at position $p$, denoted $t|_p$, is defined
recursively: $t|_\lambda = t$, $t|_k = t_k$, for $0 \leq k \leq n-1$,
and $t|_{k.p} = (t|_k)|_p$. If $t = f(t_0,\ldots,t_{n-1})$, then we
call $t_0,\ldots,t_{n-1}$ the \emph{depth-$1$} subterms of
$t$. If all variables appearing in $t$ are either $t$ itself or
depth-$1$ subterms of $t$, then we  say that $t$ is
\emph{shallow}. The notation $g[a]$ focuses on (any) one occurrence of subterm $a$ of term
$g$, and $s\{u \mapsto v\}$  denotes the term obtained from term $s$
by replacing all occurrences of the subterm $u$ in $s$ by term $v$.

A $substitution$ is a mapping $\sigma: X \rightarrow
\mathcal{T}(\mathcal{F},X)$ that is the identity on all but finitely
many elements of $X$. Substitutions are generally extended to a
homomorphism on $\mathcal{T}(\mathcal{F},X)$ in the following way: if
$t = f(t_1,\ldots,t_k)$, then (abusing notation) $\sigma(t) =
f(\sigma(t_1),\ldots,\sigma(t_k))$. Oftentimes, the application of a
substitution to a term is written in postfix notation. A $unifier$ of
two terms $s$ and $t$ is a substitution $\sigma$ (if it exists) such
that $s\sigma= t\sigma$. We assume familiarity with
the concept of {\em most general unifier} \cite{Terese}, which is unique up to variable
renaming and denoted by $mgu$. 

\noindent {\bf Term Rewrite Systems.} A \emph{rewrite rule} is a pair
of terms, $(l,r)$, usually written $l \rightarrow r$. For the rule $l
\rightarrow r$, the \emph{left-hand side} is $l \notin X$, and the
\emph{right-hand side} is $r$. Notice that $l$ cannot be a variable. A
rule, $l \rightarrow r$, can be applied to a term, $t$, if there
exists a substitution, $\sigma$, such that $l\sigma = t'$, where $t'$
is a subterm of $t$; in this case, $t$ is rewritten by replacing the
subterm $t' = l\sigma$ with $r\sigma$. The process of replacing the
subterm $l\sigma$ with $r\sigma$ is called a \emph{rewrite}. A
\emph{root rewrite} is a rewrite where $t' = t$. A rule $l \rightarrow
r$ is \emph{flat} (resp. shallow) if both $l$ and $r$ are flat
(resp. shallow). The rule $l \rightarrow r$ is \emph{collapsing} if
$r$ is a variable. A \emph{term rewrite system} (or \emph{TRS}) is a pair,
$(\mathcal{T},R)$, where $R$ is a finite set of rules and
$\mathcal{T}$ is the set of terms over some signature. A TRS, $R$, is
\emph{flat} (resp. shallow) if all of the rules in $R$ are flat
(resp. shallow). If we think of $\rightarrow$ as a relation, then
$\rightover{+}$ and $\rightover{*}$ denote its transitive
closure, and reflexive and transitive closure, respectively. Also,
$\leftrightarrow$, $\leftrightover{+}$, and $\leftrightover{*}$ 
denote the symmetric closure, symmetric and transitive closure, and
symmetric, transitive, and reflexive closure, respectively. We 
put an `r' over arrows to denote a root rewrite, i.e.,
$\leftrightover{r}$. 

A \emph{derivation} is a sequence of terms, $t_1,\ldots,t_n$, such
that $t_i \rightarrow t_{i+1}$ for $i = 1,\ldots,n-1$; this sequence
is often denoted by $t_1 \rightarrow t_2 \rightarrow \ldots
\rightarrow t_n$. A \emph{proof} is a sequence, $t_1,\ldots,t_n$, such
that $t_i \leftrightarrow t_{i+1}$ for $i = 1,\ldots,n-1$; this
sequence is generally denoted by $t_1 \leftrightarrow t_2
\leftrightarrow \ldots \leftrightarrow t_n$. If $R$ is a rewrite
system, then a proof is \emph{over} $R$ if it can be
constructed using rules in $R$. If $\pi$ is a proof, we say that $\pi
\in s \leftrightover{*} t$ if $\pi$ is of the form $s \leftrightarrow
\ldots \leftrightarrow t$ (it is possible for the proof sequence to consist 
of a single term, in which case $s = t$),
$s$). We say that $\pi \in s \leftrightover{+} t$ if $\pi \in s
\leftrightover{*} t$ and the proof sequence contains at least one step. We write $s
\leftrightover{*} t$ (resp. $s \leftrightover{+} t$) to denote that
there is a proof, $\pi$, with $\pi \in s \leftrightover{*} t$
(resp. $\pi \in s \leftrightover{+} t$). 

A \emph{normal form} is a term, $t \in \rwtermss$, such that no subterm
of $t$ can be rewritten. A term that is not a normal form, i.e., one
with a subterm that \emph{can} be rewritten, is called
\emph{reducible}. We denote the set of all normal forms for $R$
by $\NF$, or simply $NF$. A rewrite system $R$ is $UN^=$
if it is \emph{not} the case that $R$  has two distinct normal forms,
$M$ and $N$, such that $M \leftrightover{*} N$. If such a pair exists,
then we say that the pair, $\langle M,N \rangle$, is a \emph{witness}
to non-$UN^=$. The \emph{size of a witness}, denoted $|\langle M,N
\rangle|$, is $|M|+|N|$. A \emph{minimal witness} is a witness with
minimal size. Finally, we define $SubMinWit_R$ to be set of all terms
$M'$ such that $\langle M,N \rangle$ is a minimal witness, and $M'$ is
a subterm of $M$.

%
%

\section{Preliminary Results}
We begin with a few simple results on when rules apply. They are 
used throughout the paper to show that normal
forms are preserved under certain transformations. Before we begin,
notice that it is relatively simpler to preserve normal forms when the
relevant TRS is linear. For instance, imagine any {\em flat and linear} TRS 
such that $f(g(a),h(b))$ is a normal form. Since $g(a)$ is evidently a
normal form, 
$f(g(a),g(a))$ would also be a normal form, when the TRS is 
linear. If the TRS is not linear, then there could be a rule of the
form $f(x,x) \rightarrow t$, making $f(g(a),g(a))$ reducible. The
results below handle such complications presented by
non-linear rules.   

\begin{defi}

Let $R$ be a rewrite system, and let $l \rightarrow r = \rho \in R$ be
a rule. The \emph{pattern of $\rho$}, denoted $Patt(\rho)$, is a set
of equations $\{i = j~|~~l|_i = l|_j,l|_i,l|_j \in X\}$. 

\end{defi}

\begin{defi}
Let $t \in \rwtermss$ be a term with $root(l) = root(t)$. If $A =
\{i_1,i_2,\ldots,i_k\}$ is the set of positions that appear in
equations in $Patt(\rho)$, then the \emph{pattern of $t$ with respect
  to $\rho$}, denoted $Patt_\rho(t)$, is the set $\{i_a = i_b ~|~ t|_{i_a}
= t|_{i_b},i_a,i_b \in A\}$. 
\end{defi}
Note that $Patt_\rho(t)$ is undefined if $root(l) \neq root(t)$. 
\begin{lemma} \label{lem:apply}

Let $R$ be a flat TRS. Let $t \in \rwtermss$ be a term, and let $l
\rightarrow r = \rho \in R$ be a rule. Then $\rho$ can be applied to
$t$ at $\lambda$ if, and only if, (i) $l|_i = t|_i$ whenever $l|_i$ is
a constant, and (ii) $Patt_\rho(t)$ is defined and $Patt(\rho)
\subseteq Patt_\rho(t)$. 

\end{lemma}

\proof Assume that (i) and (ii) are satisfied. Since (i) is satisfied
and $Patt_\rho(t)$ is defined, all we have to show is that there exists
a substitution, $\sigma$, such that $l|_i\sigma = t|_i$ whenever $l|_i$
is a variable. We would like to define $x\sigma = t|_i$ whenever $l|_i
= x$, but if $l|_i = l|_j = x$, then $t|_i = l|_i\sigma = l|_j\sigma =
t|_j$, and hence it needs to be the case that $i = j \in
Patt_\rho(t)$. But if $i = j \in Patt(\rho)$ and (ii) is satisfied,
then we know that $i = j \in Patt_\rho(t)$. So, we can consistently
define $\sigma$ as above. Clearly, $l\sigma = t$, and thus $\rho$ can
be applied to $t$ at $\lambda$.

Now assume that there exists a substitution, $\sigma$, with $l\sigma =
t$. Obviously, $Patt_\rho(t)$ is defined and $l|_i = t|_i$ whenever
$l|_i$ is a constant, and so we need to show that $Patt(\rho)
\subseteq Patt_\rho(t)$. Say $i = j \in Patt(\rho)$. Then $l|_i =
l|_j$, and hence $t|_i = l|_i\sigma = l|_j\sigma = t|_j$. Therefore,
$i = j \in Patt_\rho(t)$, and $Patt(\rho) \subseteq Patt_\rho(t)$. \qed

Consider the term $f(a,x,x,g(b))$. Let's assume that it is a normal
form. We want to know if altering depth-$1$ subterms can make the term
reducible. Clearly, replacing $x$ with a constant could
\emph{potentially} make the term reducible, depending on the rules in
the rule set. But what about replacing any of the depth-$1$ subterms
with a normal form containing a fresh variable? Notice that such a
replacement could not make condition (i) of the above lemma true if it
had been false. But what if condition (i) is true and condition (ii)
is false? Could replacing a depth-$1$ subterm, or even several
depth-$1$ subterms, with terms containing fresh variables make
condition (ii) true? This question is answered by the following
proposition.  

\begin{proposition} \label{prop:properties}

Let $R$ be a flat TRS, and let $M = f(s_1,\ldots,s_m)$ be a normal form for $R$. Let $S = \{t_{i_1},\ldots,t_{i_n}\}$ be a set of normal forms, where $n \leq m$ and each term contains at least one fresh variable (relative to $M$). Further, say that $t_{i_j} \neq t_{i_{k}}$ whenever $s_{i_j} \neq s_{i_{k}}$ for all $i_j,i_k \in \{i_1,\ldots,i_n\}$. If $M'$ is what one obtains from $M$ by replacing each $s_{i_j}$ with $t_{i_j}$, then $M' \in NF_R$.

\end{proposition}

\proof We say that $M' = f(s_1',\ldots,s_m')$, where $s_q' =
\left\{\begin{array}{cc} t_q &\textrm{ if }q \in \{i_1,\ldots,i_n\} \\
s_q & \textrm{ otherwise} \end{array}\right.$. By Lemma
\ref{lem:apply} and the above observations, we simply need to
demonstrate, for an arbitrary rule $\rho \in R$, that if $Patt(\rho)
\not\subseteq Patt_\rho(M)$, then $Patt(\rho) \not\subseteq
Patt_\rho(M')$ (i.e. if $\rho$ cannot be applied to $M$, then it
cannot be applied to $M'$, making $M'$ a normal form).

So, assume that $Patt(\rho) \not\subseteq Patt_\rho(M)$. We need to
show that $s_j' \neq s_k'$ whenever $s_j \neq s_k$. We consider three
cases: (i) $s_j',s_k' \notin S$, (ii) $s_j' \in S$, $s_k' \notin S$,
and (iii) $s_j',s_k' \in S$. In case (i), $s_j = s_j'$ and $s_k =
s_k'$, so clearly $s_j' \neq s_k'$ whenever $s_j \neq s_k$. In case
(ii), $s_j'$ contains a fresh variable, whereas $s_k' = s_k$ does not,
so $s_j' \neq s_k'$. Hence, it is (vacuously) the case that $s_j' \neq
s_k'$ whenever $s_j \neq s_k$. Since case (iii) is an hypothesis, we
see that, in all cases, $s_j' \neq s_k'$ whenever $s_j \neq s_k$, and
hence $Patt_\rho(M') \subseteq Patt_\rho(M)$. Therefore, $Patt(\rho)
\not\subseteq Patt_\rho(M')$, and $M' \in NF_R$. \qed

\begin{lemma} \label{lem:violate minimality}
 
If $R$ is any TRS such that $f(t_1,\ldots,t_m) \in SubMinWit_R$, then
$t_i \leftrightover{*}_R t_j$ is impossible for $t_i \neq t_j$. This
is equivalent to saying that there is no term $s$ that is equivalent to both $t_i$ and
$t_j$ via $R$. 

\end{lemma}

\proof Let $\langle M,N \rangle$ be a minimal witness to non-$UN^=$
for $R$, and say that $f(t_1,\ldots,t_m)$ is a subterm of $N$. Assume
that the lemma is false, i.e., there is a term, $s$, such that $s
\leftrightover{*}_R t_i$ and $s \leftrightover{*}_R t_j$ with $t_i
\neq t_j$. Then we would have $t_i \leftrightover{*} t_j$. Since
$|t_i|+|t_j| < |f(t_1,\ldots,t_m)| < |M|+|N|$, we see that $\langle
t_i,t_j \rangle$ violates the minimality of $\langle M,N \rangle$, and
hence the lemma must be true. \qed

%
%

\subsection{Normal Forms Equivalent to Constants}
Let $E$ be a finite set of equations. Following the authors of
\cite{CHJ}, we extend $E$ to $\widehat{E}$ by closing under the
following inference rules: \begin{enumerate} \item $\displaystyle
  \frac{g = d\textrm{,  }l = r}{d \sigma = r \sigma}$ if $l,g \notin
  X$ and $\sigma = mgu(l,g)$ \item $\displaystyle \frac{x = d\textrm{,
    }y = r}{d = r\{y \mapsto x\}}$ if $y \in X$ and $x \in \constants
  \cup X$ \item $\displaystyle \frac{g[a] = d\textrm{,  }a = b}{g[b] =
    d}$ if $a,b \in \constants$ \end{enumerate} Notice that if $E$ is
flat, then $\widehat{E}$ is flat, as well. 

We can think of a rewrite system as a set of equations: if $s
\rightarrow t$ is a rule in $R$, then $s \leftrightarrow t$ is its
corresponding equation. We write $E_R$ for the set of equations
obtained in this way from a rewrite system $R$. Clearly, if $s$ and
$t$ are terms in $\rwtermss$, then they are $R$-equivalent if and only
if they are $E_R$ equivalent. Also, from \cite{CHJ} we know that terms
are $E_R$ equivalent if, and only if, they are
$\widehat{E_R}$-equivalent. In \cite{CHJ}, the authors show that, if
$R$ is a shallow TRS and $s,t \in \rwtermss$, then there is a procedure
that produces, for any proof, $\pi \in s ~\leftrightover{*}_R~ t$,
over $R$, a new proof, which is denoted by $\pi_{1rr} \in s
~\leftrightover{*}_{\widehat{E_R}}~ t$, over $\widehat{E_R}$, such
that there is at most one root rewrite step in $\pi_{1rr}$. 


Consider the following example: $R = \{f(x,x) \rightarrow c, f(x,x)
\rightarrow g(a,x), g(a,x) \rightarrow g(a,x), a \rightarrow h(b), b
\rightarrow h(c)\}$. It is easy to check that $\widehat{E_R} = E_R
\cup \{c \leftrightarrow g(a,x)\}$. We  use $\widehat{E_R}$ to
search for a minimal witness to non-$UN^=$ for $R$; in particular, we
will use the fact that for every proof $s \leftrightover{*}_R t$,
there is a proof $s \leftrightover{*}_{\widehat{E_R}} t$  with at most
one root rewrite.   

Clearly, $c$ is an $R$-normal form, so if we are looking for a minimal
witness to non-$UN^=$ for $R$, $\langle c,? \rangle$ might be a good
first guess. We know that $c \leftrightarrow_{\widehat{E_R}} f(x,x)$,
so maybe $\langle c,f(u,v) \rangle$ is a minimal witness, for some
normal forms $u$ and $v$. This is not possible. First, notice that
$f(x,x)$ appears on the LHS of a rule, so $f(t,t)$ cannot be a normal
form, for arbitrary term $t$. Second, notice that if $f(t,t)$ is
equivalent to another normal form, then we can assume it is of the
form $f(u,v)$, because we have already ``used up'' our only root
rewrite by using $c \leftrightarrow_{\widehat{E_R}} f(x,x)$. So, maybe
we can plug some term, $t$, into $x$, and then rewrite one instance of
it to a normal form $u$, and another instance of it to a normal form
$v$, obtaining a minimal witness of the form $\langle c,f(u,v)
\rangle$? This cannot be the case, because if $\langle c,f(u,v)
\rangle$ is a minimal witness, then (by Lemma \ref{lem:violate
  minimality} and the fact that $u \leftrightover{*} v$) $\langle u,v
\rangle$ would violate the minimality of $\langle c,f(u,v)
\rangle$. So, we should consider $c \leftrightarrow_{\widehat{E_R}}
g(a,x)$ as \emph{the} (one and only) rewrite step in our proof. We
know that $a$ is not a normal form, and must, therefore, be rewritten
to one - $h(h(c))$. But what about $x$? Should we plug anything into
it? Say we were to plug $t$ into $x$, and then rewrite $t$ to some
normal form, $u$. This would be unnecessary, because non-linearity is
not an issue here, and so we can leave $x$ as it is. So, $\langle
c,g(h(h(c)),x) \rangle$ is a minimal witness, and the relevant proof
looks like: $c \leftrightarrow_{\widehat{E_R}} g(a,x)
\leftrightarrow_{\widehat{E_R}} g(h(b),x)
\leftrightarrow_{\widehat{E_R}} g(h(h(c)),x)$. 

Now, here is the interesting part. Notice that we have \emph{four}
$R$-normal forms equivalent to constants, but only \emph{three}
constants in $R$, i.e, $c \leftrightover{*}_{\widehat{E_R}} c$, $h(c)
\leftrightover{*}_{\widehat{E_R}} b$, $h(h(c))
\leftrightover{*}_{\widehat{E_R}} a$, and $g(h(h(c)),x)
\leftrightover{*}_{\widehat{E_R}} c$. From the Pigeonhole Principle,
we can conclude that there must be some constant in $R$ that is
equivalent to two distinct normal forms (of course, we already knew
this, but in general this technique will be useful). We 
generalize the lessons learned from this example in the following
results. 

\begin{lemma} \label{lem:no big plug in} Let $R$ be a flat TRS. Let
  $\langle M_0,M_1 \rangle$ be a minimal witness to non-$UN^=$ for
  $R$, and say $M = f(t_1,\ldots,t_m)$ is a subterm of $M_0$. Let $c$
  be a constant, and let $c \leftrightover{r}_{\widehat{E_R}}
  f(s_1,\ldots,s_m)\leftrightover{*}_{\widehat{E_R}}
  f(t_1,\ldots,t_m) = M$ be a proof with a single root rewrite. If
  $s_i$ is not a constant, then $height(t_i) = 0$. \end{lemma}  

\proof Let $S_{const}$ be the set of positive integers, $i$, such that
$s_i \in \constants$. If none of the $s_i$'s is a variable, then there
is nothing to show; so, assume at least one of the $s_i$'s is a
variable. Now, let \[s_j' = \left\{\begin{array}{cc} s_j &\textrm{ if
}j \in S_{const} \\ x_{s_j} &\textrm{ otherwise}
\end{array}\right. ~~\textrm{ and } \qquad t_j' =
\left\{\begin{array}{cc} t_j &\textrm{ if }j \in S_{const} \\ x_{s_j}
&\textrm{ otherwise} \end{array}\right.\] where $x_{s_j}$ is a fresh
variable not appearing in $M_0$ or $M_1$, and $x_{s_i} = x_{s_j}$ if
and only if $s_i = s_j$. We show that (i) $f(s_1',\ldots,s_m')
\leftrightover{*}_{\widehat{E_R}} f(t_1',\ldots,t_m')$, (ii)
$f(t_1',\ldots,t_m') \in NF_R$, and (iii) for $i \notin S_{const}$,
$height(t_i) = 0$. 

\emph{Part (i)}. If $j \notin S_{const}$, then $s_j' = t_j' =
x_{s_j}$. So, say $j \in S_{const}$. In this case, $s_j' = s_j
\leftrightover{*}_{\widehat{E_R}} t_j = t_j'$. So,
$f(s_1',\ldots,s_m') \leftrightover{*}_{\widehat{E_R}}
f(t_1',\ldots,t_m')$. \emph{Part (ii)}. Let $j,j' \notin S_{const}$,
and say $t_j \neq t_{j'}$. In order to apply Proposition
\ref{prop:properties}, we need to show that $t_j' \neq t_{j'}'$. From
Lemma \ref{lem:violate minimality}, we know that $t_j 
\leftrightover{*}_{\widehat{E_R}} s_j \neq s_{j'}
\leftrightover{*}_{\widehat{E_R}} t_{j'}$,
and hence $t_j' = x_{s_j} \neq x_{s_{j'}} = t_{j'}'$. Therefore, we
can apply Proposition \ref{prop:properties} to obtain that
$f(t_1',\ldots,t_m') \in NF_R$. \emph{Part (iii)}. Notice that, by (i)
and $f(s_1',\ldots,s_m') \leftrightover{*}_{\widehat{E_R}} c$
, we have $f(t_1,\ldots,t_m)
\leftrightover{*}_{\widehat{E_R}} c \leftrightover{*}_{\widehat{E_R}}
f(t_1',\ldots,t_m') = N$. Also, since $N$ contains at least one fresh
variable not appearing in $M_0$ or $M_1$, we know that $M \neq N$ and
$C[N] \neq M_0$ or $M_1$, where $C[]$ is a context and $M_0 = C[M]$. Hence $\langle C[N],M_1 \rangle$ is a
witness to non-$UN^=$, with $|C[N]| \leq |M_0|$. But $\langle M_0,M_1
\rangle$ is a minimal witness, so $|C[N]| = |C[M]|$ and $|N| =
|M|$. Since $|t_i'| = 1$ for all $i \notin S_{const}$, it must be the
case that $|t_i| = 1$. Thus, we have that $height(t_i) = height(t_i')
= 0$ for all $i \notin S_{const}$. \qed  

\begin{corollary} \label{cor:no big plug in} Under the same assumptions
  as Lemma~\ref{lem:no big plug in} plus the assumption that at least
  one of the $s_i$'s is a constant, 
  there is a $j$ such that $s_j \in \constants$ and $height(t_j) =
  height(f(t_1,\ldots,t_m)) - 1$ with $1 \leq j \leq m$. \end{corollary} 

\proof Since $height(t_i) = 0$ whenever $s_i \notin \constants$, we know that $height(t_i) \leq height(t_j)$ whenever $s_i \notin \constants$ and $s_j \in \constants$. So, amongst the direct subterms of $f(t_1,\ldots,t_m)$ with maximal height, there must be one, $t_j$, such that $s_j \in \constants$. \qed

\begin{proposition} \label{prop:small witness} Let $R$ be a flat TRS, and let $c \in \constants$. Let $\langle M,N \rangle$ be a minimal
witness, and let $N'$ be a subterm of $N$ such that $height(N') = k$. Further, let $\pi \in c \leftrightover{*} N'$ be a proof
over $R$. Then we can find either (i) $1+k$ distinct normal forms equivalent to constants, the normal forms having heights $0,1,\ldots,k$, or (ii) a witness, $\langle N_0,N_1 \rangle$, to non-$UN^=$, such that $N_0$ and $N_1$ are flat. \end{proposition}

\proof We proceed by induction on $height(N')$. For the base case we assume that $height(N') = 0$. If the proof is trivial, i.e.,
if $c = N'$, then we have $1 = 1+height(N')$ normal form ( with height zero) equivalent to a constant. So, assume that $\pi$ has at least one step.

We know that there is a proof, $\pi_{1rr} \in c \leftrightover{+}_{\widehat{E_R}} N'$, such that there is only one root rewrite
step in $\pi_{1rr}$. Since the first step in $\pi_{1rr}$ is necessarily a root rewrite, $\pi_{1rr}$ must have the form $c
\leftrightover{r} w\sigma = N'$, where the rule applied is $c \rightarrow w$ or $w \rightarrow c$, and $height(w) = 0$ (notice that if $c 
\leftrightover{r} u \leftrightover{*} N'$ for some term $u$ with $height(u) > 0$, then we would need a second root rewrite to get
back to $N'$). If $w \in X$, then $x \leftrightarrow c \leftrightarrow y$, where $x,y$ are distinct variables. Therefore, $\langle
x,y \rangle$ is a witness to non-$UN^=$ with $x$ and $y$ flat. If $w \in \constants$, then we have found $1 = 1+height(N')$ normal
form (with height zero) equivalent to a constant.

For the inductive step, assume that $height(N') > 0$, and that the proposition holds for any height strictly less than
$height(N')$. Now, $\pi_{1rr}$ has the form \[ c \leftrightover{r}_{\widehat{E_R}} f(t_1,\ldots,t_m)
\leftrightover{*}_{\widehat{E_R}} f(u_1,\ldots,u_m) = N' \] and $t_i \leftrightover{*}_{\widehat{E_R}} u_i$ for $1 \leq i \leq m$.
We have two cases: (i) there is an $i$ such that $t_i \in \constants$, and (ii) there is no such $i$. For (i), by Corollary
\ref{cor:no big plug in}, there exists an $i$ such that $t_i$ is a constant and $height(u_i) = k-1$. So, we can apply the inductive hypothesis to conclude 
that we have either (i) $1+(1+(height(N')-1)) = 1+height(N')$ distinct normal forms, with heights $0,1,\ldots,height(N')$, equivalent to constants
(the first $height(N')-1$ normal forms come from the inductive hypothesis, and the final normal form is $N'$ itself, which is equivalent to $c$), or (ii) a witness, $\langle N_0,N_1 \rangle$, to non-$UN^=$, such that $N_0$ and $N_1$ are flat.


In case (ii), if $c \leftrightarrow_{\widehat{E_R}} f(s_1,\ldots,s_m)$ is the rule used for $c \leftrightarrow_{\widehat{E_R}} f(t_1,\ldots,t_m)$, then $s_i$ is a variable for $1 \leq i \leq m$. We need to show that $f(s_1,\ldots,s_m) \in NF_R$. From Lemma \ref{lem:violate minimality}, we know that $t_i \neq t_j$ whenever $u_i \neq u_j$ for $1 \leq i,j \leq m$. Since $t_i \neq t_j$ implies that $s_i \neq s_j$, we see that $s_i \neq s_j$ whenever $u_i \neq u_j$. We can assume that the variables $s_1,\ldots,s_m$ are fresh relative to $f(u_1,\ldots,u_m)$, and so we can replace $u_i$ with $s_i$ in $f(u_1,\ldots,u_m)$, obtaining $f(s_1,\ldots,s_m) \in NF_R$ by Proposition \ref{prop:properties}. Since $f(s_1,\ldots,s_m)$ is a normal form, we can replace the variables appearing in $f(s_1,\ldots,s_m)$ with fresh variables to produce a new normal form, $f(s_1',\ldots,s_m')$, such that $f(s_1,\ldots,s_m) \leftrightarrow_{\widehat{E_R}} c \leftrightarrow_{\widehat{E_R}} f(s_1',\ldots,s_m')$. So, $\langle f(s_1,\ldots,s_m),f(s_1',\ldots,s_m') \rangle$ is our witness with $f(s_1,\ldots,s_m)$ and $f(s_1',\ldots,s_m')$ flat. \qed

\begin{corollary} \label{cor:small witness} Let $R$ be a flat TRS, and let $c \in \constants$. Let $\langle M,N \rangle$ be a minimal witness, and let $N'$ be a subterm of $N$, with $height(N') \geq |\constants|$. Further, let $\pi \in c \leftrightover{*}_R N'$ be a proof over $R$. Then we can find either (i) a witness, $\langle M_0,M_1 \rangle$, to non-$UN^=$, such that $M_0$ and $M_1$ are flat, or (ii) a witness, $\langle N_0,N_1 \rangle$, to non-$UN^=$, such that $height(N_0),height(N_1) \leq |\constants|$. \end{corollary}

\proof By Proposition \ref{prop:small witness}, we know that we can find either (a) a witness, $\langle M_0,M_1 \rangle$, to non-$UN^=$, such that $M_0$ and $M_1$ are flat, or (b) $1+height(N')$ distinct normal forms equivalent to constants. If (a) is the case, then we are done. So assume that (b) is true. Since there are $1+height(N') > |\constants|$ normal forms equivalent to, at most, $|\constants|$ constants, we know, by the Pigeonhole Principle, that a single constant is equivalent to two distinct normal forms.
From the above observation, we know that the normal forms have heights
$0$, $1$, $2$, $\ldots$, $height(N')$. The smallest (height-wise)
$1+|\constants|$ normal forms each have height no more than
$|\constants|$. So, we know that we can find a witness, $\langle
N_0,N_1 \rangle$, to non-$UN^=$, such that $height(N_0),height(N_1)
\leq |\constants|$. \qed 

\begin{proposition} \label{prop:main} Let $R$ be a flat TRS. Then, either (i)
  there does not exist a constant $c \in \constants$ and normal form $N \in SubMinWit_R$
  such that $c ~\leftrightover{*}_{\widehat{E_R}}~ N$ and $height(N)
  \geq 
  |\constants|$, or (ii) there exists a witness, $\langle N_0,N_1
  \rangle$ to non-$UN^=$ for $R$ such that $height(N_0), height(N_1)
  \leq k = max\{1,|\constants|\}$. Further, there is an effective
  procedure to decide whether (i) or (ii) is the case. \end{proposition} 

\proof Consider all ground\footnote{As in~\cite{zv06,julian}, for
  nonlinear rewrite systems also we can
  expand the signature of the rewrite system with $3\alpha$ new
  constants, where $\alpha$ is the maximum arity of a function symbol
  in the rules, and focus on ground normal forms.} normal forms over the
signature of the rewrite system, i.e., consisting of constants and
function symbols appearing in the finitely many rules of $R$, with
height less than, or equal to, $k$; we use $NF_{\leq k}$ to
denote this set. Notice that if there is a constant, $c \in
\constants$, and an element of $SubMinWit_R$, $N$, with $height(N)
\geq 
|\constants|$, such that $c \leftrightover{*} N$, then by Corollary
\ref{cor:small witness} there is a witness, $\langle N_0,N_1 \rangle$,
to non-$UN^=$ for $R$ with $height(N_0),height(N_1) \leq k$. By a
result in \cite{CHJ}, the word problem is decidable for flat
systems. So, we can construct the set of all pairs, $(s,t)$, such that
$s,t \in NF_{\leq k}$ and $s \leftrightover{*}_R t$. If we do not find
a witness to non-$UN^=$ in $NF_{\leq k}$, then we know that there is
no $c \in \constants$ and $N \in SubMinWit_R$ such that $height(N)
\geq |\constants|$ and $c \leftrightover{*}_R N$. Otherwise, we have
found the witness $\langle N_0,N_1 \rangle$ with $height(N_0),
height(N_1) \leq k$. \qed 

%
%

\subsection{Shrinking Witnesses}
Say $\langle f(a,g(b,f(c,x))),h(y,y,h(a,b,c)) \rangle$ is a witness to
non-$UN^=$ for some TRS. Can we replace big subterms of a component of
the witness, without changing the fact that it is a witness, i.e., if
we replace $g(b,f(c,x))$ with a variable, $z$, will $\langle
f(a,z),$ $h(y,y,h(a,b,c)) \rangle$ still be a witness? We show that we
can replace depth-$1$ subterms that are \emph{not} equivalent to a
constant with a variable. This shrinks the size of the witness; in
particular, only depth-$1$ subterms of such a shrunk witness that are
equivalent to a constant can have height greater than, or equal to,
the number of constants in the TRS. So, a shrunk minimal witness 
either has small components, or there is a large subterm of a
component of a minimal witness that is equivalent to a constant. If
the latter is the case, then we know, by Corollary \ref{cor:small
  witness}, that there is a small witness. 

\begin{defi} \label{defn:variable}

Let $R$ be a rewrite system. Say $X$ contains, for each term (up to
renaming of variables), $t$, a variable $x_{\overline{t}}$, where
$x_{\overline{s}} = x_{\overline{t}}$ if, and only if, $s
~\leftrightover{*}_R~ t$. Let $t = f(t_1,\ldots,t_n)$ be a term in
$\rwtermss$. Then, we define $\phi(t)$ as: \[ \phi(t) = \left\{ \begin{array}{ccc}
  x_{\overline{t}} &\textrm{ if }t\textrm{ is not equivalent to a
    constant} \\ t & \textrm{otherwise} \end{array} \right. \] 

Let $u = f(u_1,\ldots,u_m)$ for $m > 0$ and $v \in X$. We define the function $\alpha$ that maps terms to terms as follows:  $\alpha(u) = f(\phi(u_1),\ldots,\phi(u_m))$ and $\alpha(v) = v$.

\end{defi}

Notice that $\alpha(c) = c$ for $c \in \constants$, since $\alpha$ only affects depth-$1$ subterms.

\begin{lemma} \label{lem:non-root step} Let $R$ be a flat TRS, and let
  $u \leftrightarrow_R v$ be a proof over $R$, where $u
  \leftrightarrow_R v$ is not a root rewrite. Then, there is a proof 
  $\alpha(u) ~\leftrightover{*}_R~ \alpha(v)$. \end{lemma}  

\proof Say $u = f(u_1,\ldots,u_m)$ and $v = f(v_1,\ldots,v_m)$ (notice
that if $u \leftrightarrow_R v$ is not a root rewrite, then neither
$u$ nor $v$ can have height zero). Since the rewrite is not a root
rewrite, we know that there are $u_i$ and $v_i$ such that $u_i
\leftrightarrow_R v_i$, and $u_j = v_j$ for all $j \neq i$. If
$u_i,v_i$ are equivalent to a constant, then $\phi(u_i) = u_i$ and
$\phi(v_i) = v_i$, and hence $\alpha(u) \leftrightarrow_R
\alpha(v)$. If $u_i,v_i$ are not equivalent to a constant, then
$\phi(u_i) = x_{\overline{u_i}} = x_{\overline{v_i}} = \phi(v_i)$, and
hence $\alpha(u) = \alpha(v)$. \qed

\begin{lemma} \label{lem:root step} Let $R$ be a flat TRS, and let $u
  \leftrightarrow_R v$ be a proof over $R$, where $u \leftrightarrow_R
  v$ is a root rewrite. If the rewrite has the form $u = w\sigma
  \rightarrow x\sigma = v$ (i.e. it uses a collapsing rule $w \rightarrow x$), then
  $\alpha(u) \leftrightarrow_R \phi(v)$; otherwise $\alpha(u)
  \leftrightarrow_R \alpha(v)$. \end{lemma} 

\proof In case of a collapsing rule, any instantiations of $x$ appearing as depth-$1$ subterms of
$u$ are equal to $v$, and so they are replaced by $\phi(v)$ in
$\alpha(u)$. Since constants in $w$ are never replaced,  $\alpha(u)
\leftrightarrow_R \phi(v)$. Otherwise, if $s$ is a depth-$1$ subterm
of $u$ or $v$ that is an instantiation 
of a shared variable, then every depth-$1$ instance of $s$ is replaced
by $\phi(s)$ in $\alpha(u)$ and $\alpha(v)$. So, $\alpha(u)
\leftrightarrow_R \alpha(v)$. \qed 

\begin{proposition} \label{prop:flatten} Let $R$ be a flat TRS. Let $s$ and
  $t$ be terms not equivalent to a constant and  $\pi \in s
  \leftrightover{*} t$ be a proof over $R$. Then, either there is a
  proof $\alpha(s) ~\leftrightover{*}_{\widehat{E_R}}~ y$ for some
  variable $y$, or there is a proof $\alpha(s)
  ~\leftrightover{*}_{\widehat{E_R}}~ \alpha(t)$. \end{proposition} 

\proof We know that there is a proof, $\pi_{1rr}$, over
$\widehat{E_R}$ with at most one root rewrite. If $\pi_{1rr}$ has zero
steps, then $\alpha(s) = \alpha(t)$, and so $\alpha(s)
~\leftrightover{*}_{\widehat{E_R}}~ \alpha(t)$.  
Assume that $\pi_{1rr}$ has at least one step, and say that it has the
form $s = s_0 \leftrightarrow_{\widehat{E_R}} \ldots
\leftrightarrow_{\widehat{E_R}} s_k = t$ for some $k \geq 1$. We
consider three cases: (i) $\pi_{1rr}$ has no root rewrite; (ii) the
only root rewrite in $\pi_{1rr}$ uses a collapsing rule; and (iii) the
only root rewrite in $\pi_{1rr}$ does not use a collapsing rule. 

In cases (i) and (iii), we know, by lemmas \ref{lem:non-root step} and
\ref{lem:root step}, that there is a proof  $\alpha(s_i)$
$~\leftrightover{*}_{\widehat{E_R}}~$ $\alpha(s_{i+1})$ for $0 \leq i
\leq k-1$. Therefore, there is a proof $\alpha(s)
~\leftrightover{*}_{\widehat{E_R}}~ \alpha(t)$. 

In case (ii), let $w\sigma = s_j \leftrightarrow_{\widehat{E_R}}
s_{j+1} = x\sigma$ be the instance of the collapsing rule, $w
\rightarrow x$, for some $0 \leq j \leq k-1$. For $i < j$, we know
that there is a proof $\alpha(s_i)
~\leftrightover{*}_{\widehat{E_R}}~ \alpha(s_{i+1})$. By Lemma
\ref{lem:root step}, we know that $\alpha(s_j)
\leftrightarrow_{\widehat{E_R}} \phi(s_{j+1})$, and so there is a
proof $\alpha(s) ~\leftrightover{*}_{\widehat{E_R}}~
\phi(s_{j+1})$. Since the terms in $\pi_{1rr}$ cannot be equivalent to
a constant (since $s,t$ are not equivalent to a constant), we know
that $\phi(s_{j+1}) = x_{\overline{s_{j+1}}}$, and so the proof is
complete \qed 

\begin{remark} \label{rem:phi(v)_fresh} As mentioned above, for any term
  $v$ not equivalent to a constant, $\phi(v)$ can be chosen so that it
  does not appear as a subterm of any finite number of
  terms. Therefore, $\phi(s_{j+1})$ can be chosen so that it does not
  appear as a subterm of $s_0,s_1,\ldots,s_k$. We can always choose a
  fresh variable that does not appear in a finite set of terms. \end{remark} 

\begin{proposition} \label{prop:flatten still NF} Let $R$ be a flat TRS, and
  let $\langle M,N \rangle$ be a minimal witness to non-$UN^=$ for
  $R$, with $M,N$ not equivalent to a constant. Then either $\langle
  \alpha(M),y \rangle$ or $\langle \alpha(M),\alpha(N) \rangle$ is a
  witness for some variable, $y$. \end{proposition}  

\proof We know from Proposition \ref{prop:flatten} that either there is a
proof $\alpha(M) ~\leftrightover{*}_{\widehat{E_R}}~ y$ for some
variable $y$, or there is a proof 
$\alpha(M) ~\leftrightover{*}_{\widehat{E_R}}~ \alpha(N)$. So, we need
to show that (i) $\alpha(M)$, $\alpha(N)$, 
and $y$ are normal forms, and that (ii) $\alpha(M) \neq y$ (whenever
$\alpha(M) ~\leftrightover{*}_{\widehat{E_R}}~ y$) and $\alpha(M) \neq
\alpha(N)$. 

For (i), we need to show that if $s$ and $t$ are depth-$1$ subterms of
$M$ (or $N$) that are not equivalent to constants, then $\phi(s) \neq
\phi(t)$ whenever $s \neq t$. So, say that $s \neq t$. If $s
\leftrightover{*}_{\widehat{E_R}} t$, then $\langle s,t \rangle$ would
violate the minimality of $\langle M,N \rangle$, since $|s|+|t| < |M|
\leq |M|+|N|$. So, we know that $s$ and $t$ are not equivalent, and
hence $\phi(s) \neq \phi(t)$. We know by Proposition
\ref{prop:properties} that $\alpha(M)$ and $\alpha(N)$ are normal
forms, because the variables replacing subterms of $M$ and $N$ can be
chosen so that they are fresh. Since variables are always normal
forms, we know that $\alpha(M)$, $\alpha(N)$, and $y$ are normal
forms. 

For (ii), if $M$ is not a variable, then $\alpha(M)$ is not a
variable, and hence $\alpha(M) \neq y$. If $M$ is a variable, then, by
Remark \ref{rem:phi(v)_fresh}, we can choose $y$ so that it does not
appear as a subterm of $M$. So, $\alpha(M) = M \neq y$. 

To see that $\alpha(M) \neq \alpha(N)$, we need to consider two cases. If $root(M) \neq root(N)$, then clearly $\alpha(M) \neq \alpha(N)$, since $\alpha$ does not affect the outermost function symbol. If $root(M) = root(N)$, then it must be the case that $M|_i \neq N|_i$ for some integer, $i$. In order for $
\alpha(M) = \alpha(N)$ to be true, $M|_i$ and $N|_i$ must be replaced by the same variable. But this only happens when $M|_i$ and $N|_i$ are equivalent, and if $M|_i$ and $N|_i$ were equivalent, then (setting $M' = M|_i$ and $N' = N|_i$) $\langle M',N' \rangle$ would be a witness with $|M'| < |M|$ and $|N'| < |N|$. This would violate the minimality of $\langle M,N \rangle$, so $M|_i$ and $N|_i$ cannot be equivalent, and hence $M|_i$ and $N|_i$ must be replaced by distinct variables. Therefore, $\alpha(M) \neq \alpha(N)$. \qed

%
%

\section{Decidability for Flat and Shallow Rewrite Systems}
\begin{lemma} \label{lem:no constant then small} Let $R$ be a flat TRS, and say that there is no constant $c \in \constants$ and normal form $N' \in SubMinWit_R$ such that $c ~\leftrightover{*}_{\widehat{E_R}}~ N'$ and $height(N') \geq |\constants|$. Let $\langle M,N \rangle$ be a minimal witness to non-$UN^=$ for $R$. Then $height(\alpha(M)),height(\alpha(N)) \leq k = max \{1,|\constants|\}$. \end{lemma}

\proof We know that (i) all depth-$1$ subterms of $\alpha(M)$ and $\alpha(N)$ that are not equivalent to a constant are necessarily variables, and (ii) there is no constant $c \in \constants$ and normal form $N' \in SubMinWit_R$ such that $c ~\leftrightover{*}_{\widehat{E_R}}~ N'$ and $height(N') \geq |\constants|$. Hence, the depth-$1$ subterms of $\alpha(M)$ and $\alpha(N)$ are either (i) variables or (ii) elements of $SubMinWit_R$ with height strictly less than $|\constants|$. This means that the heights of $\alpha(M)$ and $\alpha(N)$ are at most $max \{1,|\constants|\}$. \qed

\begin{theorem} \label{thm:flat decidable} Let $R$ be a flat TRS. If there
  is a witness to non-$UN^=$ for $R$, then there exists a witness,
  $\langle N_0,N_1 \rangle$, with $height(N_0),height(N_1) \leq k =
  max\{1,|\constants|\}$. Hence $UN^=$ is decidable for $R$. \end{theorem} 

\proof By Proposition \ref{prop:main}, we know that there is either
(i) \emph{no} constant $c \in \constants$ and normal form $N' \in
SubMinWit_R$ such that $c ~\leftrightover{*}_{\widehat{E_R}}~ N'$ and
$height(N') \geq |\constants|$, or (ii) a witness, $\langle N_0,N_1
\rangle$ to non-$UN^=$ for $R$ such that $height(N_0),height(N_1) \leq
k$. Further, there is an effective procedure to decide if (i) or (ii)
is the case. 

If (ii) is the case, then we have our witness. So, assume that (i) is
the case, and let $\langle M,N \rangle$ be a minimal witness to
non-$UN^=$ for $R$. If $M$ and $N$ are equivalent to a constant, $c$, and 
$height(M), height(N) < |\constants|$, then we are done. So, we assume (without
loss of generality) that $M,N$ are not equivalent to a constant, and thus we
can apply Proposition \ref{prop:flatten}. Hence there is either a
proof $\alpha(M) ~\leftrightover{*}_{\widehat{E_R}}~ y$ for some
variable $y$, or a proof $\alpha(M)
~\leftrightover{*}_{\widehat{E_R}}~ \alpha(N)$. By Lemma \ref{lem:no
  constant then small}, we know that
$height(\alpha(M)),height(\alpha(N)) \leq k$. Hence, by Proposition
\ref{prop:flatten still NF}, either $\langle \alpha(M),y \rangle$ or
$\langle \alpha(M),\alpha(N) \rangle$ is a witness to non-$UN^=$ with
$height(\alpha(M)),height(\alpha(N)),|y| \leq k$.  

So, if there is a witness to non-$UN^=$ for $R$, then there is a
witness, $\langle N_0,N_1 \rangle$, with $height(N_0),height(N_1) \leq
k$. The following algorithm, on input $R$, determines if $R$ is
$UN^=$: Enumerate all ground normal forms over the
signature of the rewrite system, i.e., consisting of constants and
function symbols appearing in the finitely many rules of $R$, with height less than, or equal to,
$k$; say they are $N_0,\ldots,N_n$. In \cite{CHJ}, the authors show
that the word problem is decidable for shallow TRS. So, for $0 \leq i
< j \leq n$, check if $N_i ~\leftrightover{*}_{\widehat{E_R}}~
N_j$. If $N_i ~\leftrightover{*}_{\widehat{E_R}}~ N_j$ for some $0
\leq i < j \leq n$, then $R$ is not $UN^=$; otherwise, $R$ is
$UN^=$. \qed 

%
%

Now that we have shown that $UN^=$ is decidable for flat
rewrite systems, we extend this result to 
shallow rewrite systems. We do this by \emph{flattening} a shallow rewrite
system, i.e., transforming a shallow rewrite system into a flat one in
a way that preserves $UN^=$. 
\begin{theorem} 
Let $R$ be a shallow TRS. Then $UN^=$ is decidable for $R$. 
\end{theorem} 

\section{Undecidability for Linear and Left/Right-Flat Systems}

We begin by introducing a problem known to be undecidable.

\subsection{Post Correspondence Problem}

An instance $P$ of the Post Correspondence Problem (PCP) is defined as follows:
\begin{definition}
Given a finite set of tiles $\{ \langle u_i, v_i \rangle \enskip | \enskip 1 \le i \le n \}$ where $u_i,v_i$ are words under some finite alphabet $\Gamma$, we must decide whether a sequence of indices $i_1 \cdots i_k$ exists such that $u_{i_1} \cdots u_{i_k} = v_{i_1} \cdots v_{i_k}$.
\end{definition}

Given a PCP instance $P$ we consider $|P|$ to be the number of tiles defined for that instance. If a sequence of indices is meant as a candidate solution to the PCP instance, we call it a \emph{tile sequence}. We use the convention that $\Gamma^{\ast}$ refers to the words generated by the alphabet.

\subsection{Linear and Right-Flat Construction}
\label{lrfc}

We will construct a linear and right-flat TRS $\mathcal{R}$ that reduces PCP to the $UN^=$ problem between two normal forms: $0$ and $1$. Thus, if $0 \overset{\ast}{\longleftrightarrow} 1$ we violate $UN^=$ and there is a solution to $P$; otherwise, $P$ has no solution and $UN^=$ is preserved. A correct reduction implies $UN^=$ must be undecidable for this class of TRS.

Our construction will be composed of two parts. Part one will convert an arbitrary string into a pair of identical strings. The only normal form found in part one is $0$. Part two will convert an arbitrary tile sequence into a pair of strings generated by the tiles. The only normal form found in part two is $1$. Both parts can reach a solution to $P$. Thus, if a solution exists, then $0 \overset{\ast}{\longleftrightarrow} 1$.

Since strings are central to our construction we will work with a few conventions. The terms representing strings are sequences of unary symbols ended by $\emptyset$. Furthermore, strings and the terms that represent them are used interchangeably; we may refer to $a(b(\emptyset))$ as $ab$. For a string $s$ we denote its reversal $s^R$. Note that for $s = s_1 s_2$, we have $s^R = s_2^R s_1^R$. We liberally use $\gamma$ as a placeholder for the appropriate symbol in the alphabet $\Gamma$.

Our initial set of rules corresponds to part one:
\begin{equation*}
\mathcal{R}_{0} \enskip := \enskip
\left \{
\begin{aligned}
	\left .
	\begin{aligned}
	f(\gamma(x),\emptyset,\emptyset) &\rightarrow 0 \\
	f(\emptyset,x,y) &\rightarrow g(x,y)
	\end{aligned}
	\enskip
	\right |
	\begin{aligned}
	\gamma \in \Gamma
	\end{aligned}
\end{aligned}
\right \}
\end{equation*}
\begin{equation*}
\mathcal{R}_{S} \enskip := \enskip
\left \{
\begin{aligned}
	\left .
	\begin{aligned}
	f(\gamma(x),y,z) &\rightarrow f^{(\gamma)}(x,y,z) \\
	f(x,\gamma(y),\gamma(z)) &\rightarrow f^{(\gamma)}(x,y,z)
	\end{aligned}
	\enskip
	\right |
	\begin{aligned}
	\gamma \in \Gamma
	\end{aligned}
\end{aligned}
\right \}
\end{equation*}
Since we are working with equivalences, the orientation of a rule has no bearing on reachability. We use this to our advantage by simulating the rule $f(\gamma(x),y,z) \longleftrightarrow f(x,\gamma(y),\gamma(z))$. Notice the following structure:
\begin{displaymath}
	\xymatrix@R=3em@C=8em@!0{
	f(\gamma(x),y,z) \ar[]+D+/va(180) 1ex/;[dr]+L & & f(x,\gamma(y),\gamma(z)) \ar[]+D+/va(0) 1ex/;[dl]+R \\
	& f^{(\gamma)}(x,y,z) &
}
\end{displaymath}

In our construction, the superscripted version of a function symbol will have a reduced set of applicable rewrites. By making sure only two rewrites apply, these two rewrites can be considered a single rewrite. In a derivation between non-superscripted terms, rewriting to $f^{(\gamma)}$ fixes the next rewrite we perform. Therefore, we should view $\mathcal{R}_S$ as the set of rules $f(\gamma(x),y,z) \longleftrightarrow f(x,\gamma(y),\gamma(z))$.

The following lemmas concern $\mathcal{R}_0 \cup \mathcal{R}_S$ unless otherwise specified.

\begin{lemma}
$f(x,\emptyset,\emptyset) \underset{\mathcal{R}_S}{\overset{\ast}{\longleftrightarrow}} f(\emptyset, y, z)$ iff $x,y,z \in \Gamma^{\ast}$. \label{f:form}
\end{lemma}
\begin{proof}
Clearly the rules in $\mathcal{R}_{S}$ only allow the removal of symbols $\gamma \in \Gamma$.
\end{proof}

\begin{lemma}
$f(s,y,z) \underset{\mathcal{R}_S}{\overset{\ast}{\longleftrightarrow}} f(\emptyset, s^R(y), s^R(z))$ where $s \in \Gamma^{\ast}$. \label{f:reach}
\end{lemma}
\begin{proof}
We proceed by induction on the length of $s$. For $|s| = 1$, the rules in $\mathcal{R}_{S}$ suffice. Suppose our lemma holds for $|s| = n-1$. Given $s$ of length $n$, we can write $s = \gamma(s')$ for some $\gamma \in \Gamma$. The rules in $\mathcal{R}_{S}$ allow $f(\gamma(s'),y,z) \overset{\ast}{\longleftrightarrow} f(s', \gamma(y), \gamma(z))$. If we consider $y' = \gamma(y)$ and $z' = \gamma(z)$ our induction hypothesis applies and we are done.
\end{proof}

\begin{lemma}
Let $p = f(x,\emptyset,\emptyset)$ for some $x$. Let $q = f(\emptyset,y,z)$ for some $y,z$. For a pair of terms $(p,q)$ where $p \underset{\mathcal{R}_S}{\overset{\ast}{\longleftrightarrow}} q$ then:
\begin{itemize}
	\item $\nexists \, p' = f(x',\emptyset,\emptyset)$ for some $x' \ne x$ such that $p' \underset{\mathcal{R}_S}{\overset{\ast}{\longleftrightarrow}} q$
	\item $\nexists \, q' = f(\emptyset,y',z')$ for some $(y',z') \ne (y,z)$ such that $p \underset{\mathcal{R}_S}{\overset{\ast}{\longleftrightarrow}} q'$
\end{itemize}
\label{f:inject}
\end{lemma}
\begin{proof}
We can consider $\mathcal{R}_S$ to be $\{ f(\gamma(x),y,z) \longleftrightarrow f(x,\gamma(y),\gamma(z)) \}$ since we are only interested in non-superscripted terms. Let $\pi$ be a mapping from terms of the form $f(s_1,s_2,s_3)$ to $s_1^R s_2 s_1^R s_3$. Suppose $p' \underset{\mathcal{R}_S}{\overset{\ast}{\longleftrightarrow}} q$. By Lemma~\ref{f:form}, $\pi$ is well defined for $p$, $p'$, and $q$. Clearly there is no $p' \ne p$ such that $\pi(p') = \pi(p)$. However, if $p' \underset{\mathcal{R}_S}{\overset{\ast}{\longleftrightarrow}} q$ then $\pi(p') = \pi(q) = \pi(p)$ since the value is conserved under $\mathcal{R}_S$. A similar argument applies to $q'$.
\end{proof}

Informally, we can show Lemma~\ref{f:inject} holds by observing there is no choice of rewrite if $f|_2 = f|_3 = \emptyset$. Once we apply that rewrite we are presented with a series of meaningless choices: either backtrack or perform the only other rewrite. This is the case until we reach a term where $f|_1 = \emptyset$ or we get stuck on the way. The situation is the same if we start at $f|_1 = \emptyset$ and work our way toward $f|_2 = f|_3 = \emptyset$.

\begin{lemma}
$0 \overset{\ast}{\longleftrightarrow} g(y,z)$ iff $(y,z) = (s^R, s^R)$ for some $s \in \Gamma^{\ast}\backslash \{ \varepsilon \}$. \label{f:g}
\end{lemma}
\begin{proof}
There is only one rewrite applicable at each end term: $f(x,\emptyset,\emptyset) \rightarrow 0$ and $f(\emptyset,y,z) \rightarrow g(y,z)$. Thus, our proof will have the form $0 \leftarrow f(x,\emptyset,\emptyset) \overset{\ast}{\longleftrightarrow} f(\emptyset,y,z) \rightarrow g(y,z)$. We know by Lemma~\ref{f:form} that $x = s$ for some $s \in \Gamma^{\ast}$. Furthermore, $s$ cannot be the empty string due to how we constructed the rules in $\mathcal{R}_0$. By Lemma~\ref{f:reach} we know $(y,z) = (s^R, s^R)$. By Lemma~\ref{f:inject} we know $\impliedby$.
\end{proof}
The first part of our construction is concluded. The second part of our construction uses many of the same techniques.

For each tile $\langle u_i, v_i \rangle$ in $P$ let it be represented by the function symbol $t_i \! : \! 1$.  Let $n = \max(|u_i|,|v_i|)$. We create $n$ rules for each tile. Note that $\gamma^n_{u_i}$ and $\gamma^n_{v_i}$ refer to the $n$th symbol in $u_i$ and $v_i$, respectively. If $k > |u_i|$ then $\gamma^k_{u_i}$ leaves the variable unchanged (the concrete instantiation of the rule has only the variable in that position). Same for $k > |v_i|$. Here are the rules: 
\begin{equation*}
\mathcal{R}_{1} \enskip := \enskip
\left \{
\begin{aligned}
	h(t_i(x),\emptyset,\emptyset) &\rightarrow 1 \\
	h(\emptyset,x,y) &\rightarrow g(x,y)
\end{aligned}
\right \}
\end{equation*}
\begin{equation*}
\mathcal{R}_{T} \enskip := \enskip
\left \{
\begin{aligned}
	h(t_i(x),y,z) &\rightarrow h^{(i,0)}(x,y,z) \\
	h^{(i,k)}(x,\gamma_{u_i}^k(y),\gamma_{v_i}^k(z)) &\rightarrow h^{(i,k-1)}(x,y,z) \\
	h(x,\gamma_{u_i}^n(y),\gamma_{v_i}^n(z)) &\rightarrow h^{(i,n-1)}(x,y,z) \\
\end{aligned}
\right \}
\end{equation*}

The rules in $\mathcal{R}_T$ were constructed to simulate rules, much like $\mathcal{R}_S$. However, in $\mathcal{R}_T$ we fix a longer chain of rewrites so we can simulate $h(t_i(x),y,z) \overset{\ast}{\longleftrightarrow} h(x,u_i^R(y),v_i^R(z))$ for non-superscripted terms:
\begin{displaymath}
	\xymatrix@R=3em@C=8em@!0{
	h(T_i(x),y,z) \ar[]+D+/va(180) 1ex/;[dr]+L & & h(x,u_i^R(y),v_i^R(z)) \ar[]+D+/va(0) 1ex/;[dl]+R_(.5){\ast} \\
	& h^{(i,0)}(x,y,z) &
}
\end{displaymath}
For example, if $t_1 = \langle aab, bb \rangle$ then we would have a sequence of rules:
\begin{align*}
h(t_1(x),y,z) &\rightarrow h^{(1,0)}(x,y,z) \\
h^{(1,1)}(x,a(y),b(z)) &\rightarrow h^{(1,0)}(x,y,z) \\
h^{(1,2)}(x,a(y),b(z)) &\rightarrow h^{(1,1)}(x,y,z) \\
h(x,b(y),z) &\rightarrow h^{(1,2)}(x,y,z)
\end{align*}

The following lemmas concern $\mathcal{R}_1 \cup \mathcal{R}_T$ unless otherwise specified.

\begin{lemma}
$h(x,\emptyset,\emptyset) \underset{\mathcal{R}_T}{\overset{\ast}{\longleftrightarrow}} h(\emptyset, y,z)$ iff $x = t_{i_1} \cdots t_{i_n}$ and $y,z \in \Gamma^{\ast}$. \label{h:form}
\end{lemma}
\begin{proof}
Clearly, the rules in $\mathcal{R}_T$ only allow the removal of symbols $t_i$ from $h|_1$ and the removal of symbols $\gamma \in \Gamma$ from $h|_2$ and $h|_3$.
\end{proof}

\begin{lemma}
$h(t,y,z) \underset{\mathcal{R}_T}{\overset{\ast}{\longleftrightarrow}} h(\emptyset, s_a^R(y),s_b^R(z))$ where $t = t_{i_1} \cdots t_{i_n}, \enskip s_a = u_{i_1} \cdots u_{i_n}$ and $s_b = v_{i_1} \cdots v_{i_n}$. \label{h:reach}
\end{lemma}
\begin{proof}
We proceed by induction on the length of $t$. For $|t| = 1$, the rules in $\mathcal{R}_{T}$ suffice. Suppose our lemma holds for $|t| = n-1$. Given $t$ of length $n$, we can write $t = t_i(t')$ for some $t_i$. The rules in $\mathcal{R}_{T}$ allow $h(t_i(t'),y,z) \overset{\ast}{\longleftrightarrow} h(t', u_i^R(y), v_i^R(z))$. If we consider $y' = u_i^R(y)$ and $z' = v_i^R(z)$ our induction hypothesis applies and we are done.
\end{proof}

\begin{lemma}
Let $p = h(x,\emptyset,\emptyset)$ for some $x$. Let $q = h(\emptyset,y,z)$ for some $y,z$. For a pair of terms $(p,q)$ where $p \underset{\mathcal{R}_T}{\overset{\ast}{\longleftrightarrow}} q$ then:
\begin{itemize}
	\item $\nexists \, p' = h(x',\emptyset,\emptyset)$ for some $x' \ne x$ such that $p' \underset{\mathcal{R}_T}{\overset{\ast}{\longleftrightarrow}} q$
	\item $\nexists \, q' = h(\emptyset,y',z')$ for some $(y',z') \ne (y,z)$ such that $p \underset{\mathcal{R}_T}{\overset{\ast}{\longleftrightarrow}} q'$
\end{itemize}
\label{h:inject}
\end{lemma}
\begin{proof}
We can consider $\mathcal{R}_T$ to be $\{ h(t_i(x),y,z) \overset{\ast}{\longleftrightarrow} h(x,u_i^R(y),v_i^R(z)) \}$ since we are only interested in non-superscripted terms. Let $\pi_a, \pi_b$ be mappings from tile sequences $t_{i_1}\cdots t_{i_n}$ to $u_{i_1}\cdots u_{i_n}$ and $v_{i_1}\cdots v_{i_n}$, respectively. Let $\pi$ be a mapping from terms of the form $h(t_1,s_2,s_3)$ to $(\pi_a(t_1))^R s_2 (\pi_b(t_1))^R s_3$. Suppose $p' \underset{\mathcal{R}_T}{\overset{\ast}{\longleftrightarrow}} q$. By Lemma~\ref{h:form}, $\pi$ is well defined for $p$, $p'$, and $q$. Clearly there is no $p' \ne p$ such that $\pi(p') = \pi(p)$ (holds as long as $t_i \ne t_j$ for $i \ne j$). However, if $p' \underset{\mathcal{R}_T}{\overset{\ast}{\longleftrightarrow}} q$ then $\pi(p') = \pi(q) = \pi(p)$ since the value is conserved under $\mathcal{R}_T$. A similar argument applies to $q'$.
\end{proof}

The informal argument used in part one unfortunately does not
apply. Let $t_1 = \langle abb, ba \rangle$ and $t_2 = \langle bb, ba
\rangle$. We have the following equivalence with our rules:
$h(t_1,\emptyset,\emptyset)
\underset{\mathcal{R}_T}{\overset{\ast}{\longleftrightarrow}}
h(\emptyset, bba, ab)
\underset{\mathcal{R}_T}{\overset{\ast}{\longleftrightarrow}} h(t_2,
a, \emptyset)$. Although it may seem like an error, the values of all
three terms under $\pi$ are indeed the same.

\begin{lemma}
$1 \overset{\ast}{\longleftrightarrow} g(y,z)$ iff $(y,z) = (s_a^R, s_b^R)$ where $s_a = u_{i_1} \cdots u_{i_n}$ and $s_b = v_{i_1} \cdots v_{i_n}$ for some nonempty tile sequence $i_{1} \cdots i_{n}$. \label{h:g}
\end{lemma}
\begin{proof}
There is only one rewrite applicable at each end term: $h(x,\emptyset,\emptyset) \rightarrow 1$ and $h(\emptyset,y,z) \rightarrow g(y,z)$. Thus, our proof will have the form $1 \leftarrow h(x,\emptyset,\emptyset) \overset{\ast}{\longleftrightarrow} h(\emptyset,y,z) \rightarrow g(y,z)$. We know by Lemma~\ref{h:form} that $x = t$ for some $t = t_{i_1} \cdots t_{i_n}$. Furthermore, $t$ cannot be an empty sequence due to how we constructed the rules in $\mathcal{R}_1$. By Lemma~\ref{h:reach} we know $(y,z) = (s_a^R, s_b^R)$. By Lemma~\ref{h:inject} we know $\impliedby$.
\end{proof}

\begin{lemma}
$0 \overset{\ast}{\longleftrightarrow} 1$ iff $P$ has a solution. \label{soln}
\end{lemma}
\begin{proof}
Any proof of $0 \overset{\ast}{\longleftrightarrow} 1$ must go through some term $g(x,y)$. Due to Lemma~\ref{f:g} and Lemma~\ref{h:g} the only term $g(x,y)$ that both $0$ and $1$ can reach must have $x$ and $y$ as a pair of identical strings generated by the tiles in $P$. Thus, $P$ must have a solution.
\end{proof}

Finally, we add the set of rules that guarantee $0$ and $1$ are the only normal forms. These rules do not disturb any of the results above.
\begin{equation*}
\mathcal{R}_{nf} \enskip := \enskip
\left \{
\begin{aligned}
	f(x,y,z) &\rightarrow f(x,y,z) & h(x,y,z) &\rightarrow h(x,y,z) \\
	f^{(\gamma)}(x,y,z) &\rightarrow f^{(\gamma)}(x,y,z) & h^{(i,j)}(x,y,z) &\rightarrow h^{(i,j)}(x,y,z) \\
	g(x,y) &\rightarrow g(x,y) & \gamma(x) &\rightarrow \gamma(x) \\
	\emptyset &\rightarrow \emptyset & t_i(x) &\rightarrow t_i(x)
\end{aligned}
\right \}
\end{equation*}
Now that $0$ and $1$ are the only normal forms, their equivalence implies a violation of $UN^=$. Thus, our complete set of rules is: $\mathcal{R} := \mathcal{R}_0 \cup \mathcal{R}_1 \cup \mathcal{R}_S \cup \mathcal{R}_T \cup \mathcal{R}_{nf}$.

\begin{theorem}
$UN^=$ is undecidable for linear TRS that are right-flat and have left-hand sides of depth two. 
\end{theorem}
\begin{proof}
Direct consequence of Lemma~\ref{soln}, which proves our construction reduces $UN^=$ to solving PCP.
\end{proof}

\begin{figure}
\begin{align*}
0 &\lra f(bbaabbbaa,\emptyset,\emptyset) &&\lra h(\emptyset,aabbbaabb,aabbbaabb) \\
 &\lra f^{(b)}(baabbbaa,\emptyset,\emptyset) &&\lra h^{(1,2)}(\emptyset,aabbbaabb,abbbaabb) \\
 &\lra f(baabbbaa,b,b) &&\lra h^{(1,1)}(\emptyset,aabbbaabb,bbbaabb) \\
 &\lra f^{(b)}(aabbbaa,b,b) &&\lra h^{(1,0)}(\emptyset,abbbaabb,bbaabb) \\
 &\lra f(aabbbaa,bb,bb) &&\lra h(t_1,abbbaabb,bbaabb) \\
 &\lra f^{(a)}(abbbaa,bb,bb) &&\lra h^{(3,2)}(t_1,bbbaabb,bbaabb) \\
 &\lra f(abbbaa,abb,abb) &&\lra h^{(3,1)}(t_1,bbaabb,baabb) \\
 &\lra f^{(a)}(bbbaa,abb,abb) &&\lra h^{(3,0)}(t_1,baabb,aabb) \\
 &\lra f(bbbaa,aabb,aabb) &&\lra h(t_3t_1,baabb,aabb) \\
 &\lra f^{(b)}(bbaa,aabb,aabb) &&\lra h^{(2,1)}(t_3t_1,aabb,abb) \\
 &\lra f(bbaa,baabb,baabb) &&\lra h^{(2,0)}(t_3t_1,abb,bb) \\
 &\lra f^{(b)}(baa,baabb,baabb) &&\lra h(t_2t_3t_1,abb,bb) \\
 &\lra f(baa,bbaabb,bbaabb) &&\lra h^{(3,2)}(t_2t_3t_1,bb,bb) \\
 &\lra f^{(b)}(aa,bbaabb,bbaabb) &&\lra h^{(3,1)}(t_2t_3t_1,b,b) \\
 &\lra f(aa,bbbaabb,bbbaabb) &&\lra h^{(3,0)}(t_2t_3t_1,\emptyset,\emptyset) \\
 &\lra f^{(a)}(a,bbbaabb,bbbaabb) &&\lra h(t_3t_2t_3t_1,\emptyset,\emptyset) \lra 1 \\
 &\lra f(a,abbbaabb,abbbaabb) \\
 &\lra f^{(a)}(\emptyset,abbbaabb,abbbaabb) \\
 &\lra f(\emptyset,aabbbaabb,aabbbaabb) \\
 &\lra g(aabbbaabb,aabbbaabb)
\end{align*}
\caption{An example for $P = \{ \langle a, baa \rangle, \langle ab,
aa \rangle, \langle bba, bb \rangle \}.$}
\end{figure}

\subsection{Linear and Left-Flat Construction}

If we reverse the orientation of all rules in $\mathcal{R}$ we run into a small problem: the $\mathcal{R}$ normal forms $0$ and $1$ are no longer normal forms after reorientation. To remedy this, we replace the rules $\{ f(\gamma(x),\emptyset,\emptyset) \rightarrow 0, \enskip h(T_i(x),\emptyset,\emptyset) \rightarrow 1 \}$ with the following modifications:
\begin{equation*}
\mathcal{R}_{j} := \left \{
\begin{aligned}
	j_0(x) &\rightarrow 0 \\
	j_0(x) &\rightarrow f(x,\emptyset,\emptyset) \\
	j_1(x) &\rightarrow h(\emptyset,\emptyset,x) \\
	j_1(x) &\rightarrow 1
\end{aligned}
\right \}
\end{equation*}
Thus, 0 and 1 remain normal forms after reorientation and $Var(r) \subset Var(l)$. However, we must now disallow the empty string as a solution somewhere else in the construction. To that end, we replace $\{ f(\emptyset,x,y) \rightarrow g(x,y), \enskip h(\emptyset,x,y) \rightarrow g(x,y) \}$ with the following:
\begin{equation*}
\mathcal{R}_{g} := \left \{
\begin{aligned}
	g^{(\gamma,\gamma)}(x,y) &\rightarrow f(\emptyset, \gamma(x), \gamma(y)) \\
	g^{(\gamma_i,\gamma_j)}(x,y) &\rightarrow h(\emptyset, \gamma_i(x), \gamma_j(y))
\end{aligned}
\right \}
\end{equation*}
Any rules that have not been replaced are simply reoriented. Thus, our final rule set is:
\begin{align*}
\mathcal{R} \enskip := \enskip \mathcal{R}_{j} &\cup \mathcal{R}_{g} \cup \mathcal{R}_{S}^{-1} \cup \mathcal{R}_{T}^{-1} \\
	&\cup \mathcal{R}_{nf}\backslash \{ g(x,y) \rightarrow g(x,y) \} \\
	&\cup \{ g^{(\gamma_i,\gamma_j)}(x,y) \rightarrow g^{(\gamma_i,\gamma_j)}(x,y) \}
\end{align*}

\begin{theorem}
$UN^=$ is undecidable for linear TRS that are left-flat and have right-hand sides of depth two.
\end{theorem}

\begin{proof}
All proofs in Section~\ref{lrfc} can be easily adapted for this modified TRS.
\end{proof}

\section{Conclusion}

The \UNeq{} property of TRSs is shown to be decidable for the
shallow class and undecidable for the class of linear TRS in which
one side of the rule is allowed to be at most depth-two and the other side is
flat. Among the fundamental properties of TRSs only the word problem and
the \UNeq{} property are now known to be decidable for the
shallow class. An important direction for future research is to give a
complete classification of the basic properties for all subclasses of
linear, depth-two TRSs (see also \cite{rakrule08} in this regard).

\section*{Acknowledgments.} We thank Ross Greenwood and the reviewers
of FSTTCS 2010 for careful readings and constructive comments. 
%
%

\bibliographystyle{ACM-Reference-Format-Journals}

{\small \bibliography{/auto/disk01/faculty/cosc/rmverma/latex/space}}

\end{document}